\documentclass[11pt]{article}
\RequirePackage{algorithm,algorithmic,amssymb,epsf,epic}

\newtheorem{invariant}{Invariant}

\def\denseformat{
\setlength{\textheight}{9.5in}
\setlength{\textwidth}{6.9in}
\setlength{\evensidemargin}{-0.3in}
\setlength{\oddsidemargin}{-0.3in}
\setlength{\headsep}{10pt}
\setlength{\topmargin}{-0.44in}
\setlength{\columnsep}{0.375in}
\setlength{\itemsep}{0pt}
}

\newtheorem{theorem}{Theorem}[section]

\newtheorem{definition}[theorem]{Definition}

\newtheorem{lemma}[theorem]{Lemma}

\newtheorem{corollary}[theorem]{Corollary}

\newtheorem{comment}[theorem]{Comment}

\def\boldhead#1:{\par\vskip 7pt\noindent{\bf #1:}\hskip 10pt}
\def\ithead#1:{\par\vskip 7pt\noindent{\it #1:}\hskip 10pt}

\def\inline#1:{\par\vskip 7pt\noindent{\bf #1:}\hskip 10pt}
\def\midinline#1:{\par\noindent{\bf #1:}\hskip 10pt}
\def\dnsinline#1:{\par\vskip -7pt\noindent{\bf #1:}\hskip 10pt}
\def\ddnsinline#1:{\newline{\bf #1:}\hskip 10pt}
\def\largeinline#1:{\par\vskip 7pt\noindent{\large\bf #1:}\hskip 10pt}
\long\def\comment #1\commentend{}
\long\def\commhide #1\commhideend{}
\long\def\commfull #1\commend{#1}
\long\def\commabs #1\commenda{}
\long\def\commtim #1\commendt{#1}
\long\def\commb #1\commbend{}
\long\def\commedit #1\commeditend{} 

\long\def\commB #1\commBend{}       

\long\def\commex #1\commexend{}     

\long\def\commsiena #1\commsienaend{}  
                                         
\long\def\commBI #1\commBIend{}  
                                         
\long\def\CProof #1\CQED{}
\def\qed{\mbox{}\hfill $\Box$\\}
\def\blackslug{\hbox{\hskip 1pt \vrule width 4pt height 8pt
    depth 1.5pt \hskip 1pt}}
\def\QED{\quad\blackslug\lower 8.5pt\null\par}

\long\def\PPP#1{\noindent{\bf Proof:}{ #1}{\quad\blackslug\lower 8.5pt\null}}

\long\def\denspar #1\densend
{#1}

\newif\ifnotesw\noteswtrue
   {\ifnotesw\marginpar[\hfill\(\top\)]{\(\top\)}\fi}%
   {\ifnotesw\marginpar[\hfill\(\bot\)]{\(\bot\)}\fi}

\newcommand{\mnote}[1]%
    {\ifnotesw\marginpar%
        [{\scriptsize\it\begin{minipage}[t]{\marginparwidth}
        \raggedleft#1%
                        \end{minipage}}]%
        {\scriptsize\it\begin{minipage}[t]{\marginparwidth}
        \raggedright#1%
                        \end{minipage}}%
    \fi}

\def\cG{{\cal G}}

\def\MathF{\hbox{\rm I\kern-2pt F}}
\def\MathP{\hbox{\rm I\kern-2pt P}}
\def\MathR{\hbox{\rm I\kern-2pt R}}
\def\MathZ{\hbox{\sf Z\kern-4pt Z}}
\def\MathN{\hbox{\rm I\kern-2pt I\kern-3.1pt N}}
\def\MathC{\hbox{\rm \kern0.7pt\raise0.8pt\hbox{\footnotesize I}
\kern-4.2pt C}}
\def\MathQ{\hbox{\rm I\kern-6pt Q}}

\newsavebox{\ttop}\newsavebox{\bbot}

\def\eps{\epsilon}

\denseformat

\newcommand{\free}{\texttt{has-free}}

\newcommand{\gfree}{\texttt{get-free}}
\newcommand{\add}{\texttt{match}}
\newcommand{\mate}{\texttt{mate}}
\newcommand{\apath}{\texttt{aug-path}}
\newcommand{\surrogate}{\texttt{surrogate}}

\newcommand {\ignore} [1] {}

\begin{document}

\title{Simple Deterministic Algorithms for Fully Dynamic Maximal Matching}
\author{
Ofer Neiman \thanks{Department of Computer Science,
        Ben-Gurion University of the Negev, POB 653, Beer-Sheva 84105, Israel.
         \newline E-mail: {\tt neimano@cs.bgu.ac.il}.
         This work is supported by ISF grant No. (523/12) and by the European Union's Seventh Framework Programme (FP7/2007-2013) under grant agreement $n^\circ$303809.}
\and
Shay Solomon \thanks{Department of Computer Science and Applied Mathematics, The Weizmann Institute of Science, Rehovot 76100, Israel. E-mail: {\tt shay.solomon@weizmann.ac.il}. This work is supported by the Koshland Center for basic Research.
\newline Part of this work was done while the author was a graduate student at the Ben-Gurion University of the Negev,
under the support of the Clore Fellowship grant No.\ 81265410, the BSF grant No.\ 2008430, and the ISF grant No.\ 87209011.}}

\date{\empty}

\begin{titlepage}
\def\thepage{}
\maketitle

\begin{abstract}
A maximal matching can be maintained in fully dynamic (supporting both addition and deletion of edges)
$n$-vertex graphs using a trivial deterministic algorithm with a worst-case update time of $O(n)$.
No deterministic algorithm that outperforms the na\"{\i}ve $O(n)$ one was reported up to this date. The only progress in this direction is due to Ivkovi\'{c} and Lloyd \cite{IL93},
who in 1993 devised a deterministic algorithm with an \emph{amortized} update time of $O((n+m)^{\sqrt{2}/2})$, where $m$ is the number of edges.

In this paper we show the first deterministic fully dynamic algorithm that outperforms the trivial one.
Specifically, we provide a deterministic \emph{worst-case} update time of $O(\sqrt{m})$.
Moreover, our algorithm maintains a matching which is in fact a $3/2$-approximate maximum cardinality matching (MCM).
We remark that no fully dynamic algorithm for maintaining $(2-\eps)$-approximate MCM improving upon the na\"{\i}ve $O(n)$ was known prior to this work, even allowing amortized time bounds and \emph{randomization}.

For low arboricity graphs (e.g., planar graphs and graphs excluding fixed minors), we devise another simple deterministic algorithm with \emph{sub-logarithmic update time}.
Specifically, it maintains a fully dynamic maximal matching with amortized update time of $O(\log n/\log \log n)$.
This result addresses an open question of Onak and Rubinfeld \cite{OR10}.

We also show a deterministic algorithm with optimal space usage, that for arbitrary graphs maintains a maximal matching in amortized $O(\sqrt{m})$ time, and uses only $O(n+m)$ space.
\end{abstract} 
\end{titlepage}

\pagenumbering {arabic} 

\section{Introduction}

In this paper we study deterministic algorithms for maximal matching in a dynamically changing graph. While graphs have been traditionally studied as static objects, in some of the modern applications of graph theory (e.g. communication and social networks, graphics and AI) graphs are subject to discrete changes. In the last few decades there has been a growing interest in algorithms and data structures for such dynamically changing graphs. In particular, several classical combinatorial problems, such as connectivity, min-cut, minimum spanning tree \cite{EGIN97, HK99, Thor07,HLT01,AL94,KS02,DI04,DI05,AT07}, have been considered in such a dynamic setting.

Our goal is to design a data structure that maintains a maximal matching, or an approximate maximum cardinality matching, in a fully dynamic graph. This dynamic setting allows both insertions and deletions of edges, while the vertex set is fixed. A standard assumption is that in each step a single edge is added to the graph or removed from it; such a step is called an \emph{edge update}
(or shortly, an \emph{update}). Throughout the paper let $n$ denote the number of vertices in the graph, and $m$ the (current) number of edges.

Observe that a simple greedy algorithm computes a maximal matching in $O(m)$ time, so recomputing a maximal matching from scratch would cost $O(m)$ per update. It is not hard to see that one can also dynamically maintain a maximal matching with a worst-case update time of $O(n)$ (see Section \ref{sec:overview}).
To the best of our knowledge there was no deterministic algorithm known that beats this na\"{\i}ve $O(n)$ time, even allowing amortized time bounds. Ivkovi\'{c} and Lloyd \cite{IL93} showed an algorithm with an \emph{amortized} update time of $O((n+m)^{\frac{\sqrt{2}}{2}})$, which improves upon the $O(n)$ time algorithm only when $m = O(n^{\sqrt{2}})$.
\vspace{0.12in}\\
{\bf Our results.}
We show a simple deterministic algorithm that for any dynamic graph maintains explicitly a maximal matching in $O(\sqrt{m})$ \emph{worst-case} update time.\footnote{We ignore additive terms that depend at most logarithmically on $n$.} That is, if the current graph has $m$ edges, then the next add or delete operation will be handled in $O(\sqrt{m})$ time.
This improves upon the amortized bound of \cite{IL93} for all values of $m$.  Moreover, our algorithm is arguably simpler than that of \cite{IL93}.

It is well known that a maximal matching is in particular a $2$-approximation for maximum cardinality matching (MCM). Our algorithm has the additional property that there are no augmenting paths of length $3$ at all times, which implies that the matching we maintain is in fact a $3/2$-approximation for MCM. Remarkably, no algorithm with update time better than $O(n)$ was known for maintaining $(2-\eps)$-approximate MCM (for any $\eps>0$), even allowing amortized time bounds and \emph{randomization}.
Our deterministic data structure also maintains an approximate \emph{vertex cover}, as it is well-known that the set of vertices participating in a maximal matching is in fact a $2$-approximate vertex cover.

We also obtain improved bounds for low arboricity graphs, which are uniformly sparse graphs (see Definition \ref{def:arb}). We show a simple deterministic algorithm that maintains a maximal matching in amortized $O\left(\frac{\log n}{\log((\log n)/c)} + c\right)$ time per update,
provided that the dynamic graph has arboricity at most $c=o(\log n)$ at all times. When $c=\Omega(\log n)$ we obtain amortized $O(c)$ time.
It is well known that the arboricity of any graph with $m$ edges is at most $\sqrt{m}$. So this algorithm for bounded arboricity graphs gives rise to an even simpler $O(\sqrt{m})$ amortized time data structure for dynamic maximal matching in \emph{arbitrary} graphs (but not a $3/2$-approximate MCM).
Moreover, we show that this algorithm can be implemented using \emph{optimal} space $O(n+m)$.

At the ``Open Problems'' section of
\cite{OR10} the following question was raised:
``Is there a \emph{deterministic} data structure that achieves
a constant approximation factor with polylogarithmic
update time?'' Our result gives a $2$-approximation with sub-logarithmic update time for low arboricity graphs.

\subsection{Related Work}

Maintaining the maximum cardinality matching dynamically seems to be a difficult task. The state-of-the-art static algorithm by
Micali and Vazirani from 1980 \cite{MV80} takes $O(\sqrt{n}\cdot m)$ time, which suggests that obtaining a dynamic algorithm with $o(\sqrt{n})$ amortized time would be considered a breakthrough. Sankowski \cite{Sank07} showed a randomized algorithm with $O(n^{1.495})$ time per update that dynamically maintains an MCM. In a certain randomized model, where the edge being added or deleted is chosen at random, Alberts and Henzinger \cite{AH95} showed that MCM can be maintained in amortized $O(n)$ update time.

Recently some very successful randomized algorithms were developed for maintaining an approximate MCM dynamically. Onak and Rubinfeld \cite{OR10} showed a randomized $O(1)$-approximate MCM whose amortized time per update is $O(\log^2n)$ with high probability. This was improved by Baswana, Gupta and Sen \cite{BGS11} to a randomized algorithm that maintains a maximal matching (and in particular $2$-approximate MCM) in $O(\log n)$ expected amortized update time.

Concurrently and independently of our work, Anand \cite{anand} obtained similar results for deterministic fully dynamic approximate MCM. Specifically, Anand shows a $3/2$-approximate MCM in amortized $O(\sqrt{m})$ update time.
(Recall that our bound $O(\sqrt{m})$ is worst-case rather than amortized.)

\subsection{Overview of the Algorithm}\label{sec:overview}

For simplicity of  presentation, in this extended abstract we prove a somewhat weaker bound of $O(\sqrt{m+n})$ on the worst-case update time.
The proof of the bound $O(\sqrt{m})$ follows similar lines, and is deferred to the full version.

Let us recall how the na\"{\i}ve $O(n)$ time per update algorithm works. For every update, if an edge is added to the graph, check if it can be added to the matching. If a matched edge $\{u,v\}$ is deleted from the graph, examine all the neighbors of $u$ and $v$ to see if some edge $\{u,w\}$ or $\{v,w'\}$ (or both) can be added to the matching. It is not hard to verify that the resulting matching remains maximal. Now the question is, can we do anything better than scanning all the neighbors of a free vertex in order to find a new match for it? We believe that, in general, the answer is no.

The way our algorithm overcomes this obstacle is by ensuring that high degree vertices are \emph{never} free. In particular, we maintain the following invariant: vertices of degree larger than $\sqrt{2(m+n)}$ are matched at all times. Then scanning all neighbors of a free vertex is not so expensive. Next we briefly explain how to maintain this invariant. When a high degree vertex $u$ becomes free and cannot be matched
(because all its neighbors are matched), we find a \emph{surrogate} for it, that is, a vertex $v'$ that is matched to a neighbor $v$ of $u$, such that the degree of $v'$ is at most $\sqrt{2m}$. Then we can match $u$ to $v$, and the low degree vertex $v'$ becomes free instead of $u$. We prove that such a vertex $v'$ must exist, and show how to find one in $O(\sqrt{m+n})$ time.

One has to be careful when defining the invariant with respect to the number of edges, as this number changes with time. It is even possible that at some point many vertices violate the invariant simultaneously. The first attempt is to find a low degree surrogate for each of these vertices.
Finding a surrogate, however, takes $O(\sqrt{m+n})$ time, and so we cannot handle many vertices at once.
Instead, at each edge update we handle $O(1)$ ``problematic'' vertices, those that are getting close to violating the invariant. By handling the problematic vertices in decreasing order of degree (one at each edge update), we demonstrate that each problematic vertex will be handled
long before it can violate the invariant.

In order to obtain a $3/2$-approximate MCM, this approach does not suffice. When a vertex $u$ becomes free but has no free neighbors, we may be forced to search every one of its $\sqrt{2m}$ neighbors $v$, who are matched to $v'$, for a free neighbor of $v'$. To this end we use another idea: instead of searching for a free neighbor, we maintain a certain data structure for every vertex that holds information about all its free neighbors. This data structure will enable us to
determine the existence of a free neighbor in $O(1)$ time, update a single neighbor (a free neighbor that becomes matched, and vice versa) in $O(1)$ time, and find a free neighbor in $O(\sqrt{n})$ time. Observe that whenever a vertex changes its status (free or matched) it must inform all of its neighbors in order to update their free neighbors data structure. However, since we guarantee that high degree vertices never change their status (they are always matched),
 updating this data structure will only cost $O(\sqrt{m+n})$ time per update.
\vspace{0.12in}\\
{\bf Low arboricity graphs.}
For graphs with bounded arboricity we use a result of Brodal and Fagerberg \cite{BF99}, who devised a fully dynamic algorithm for maintaining a bounded \emph{edge orientation}. That is, assign a direction for each edge in the graph such that the out-degree of every vertex is bounded. Although the na\"{\i}ve $O(n)$ time algorithm mentioned above is in fact very efficient for bounded degree graphs (it runs in $O(d)$ time if $d$ is the maximum degree), in bounded arboricity graphs the in-degree can be arbitrarily high. Our algorithm does the following: vertices have only \emph{partial information} about their free neighbors. In particular, each vertex will send information about its status only along the out-going edges of the orientation. This greedy approach guarantees that each vertex will hold authentic information about all its (possibly many) in-coming neighbors
at all times. Information about the few out-going neighbors will not be authentic, but can be verified on demand by scanning all of them.

\section{Preliminaries}

Let $G = (V,E)$ be an arbitrary graph.
Any set $M \subseteq E$ of vertex-disjoint edges is called a \emph{matching}.
A matching of maximum cardinality in $G$ is called a \emph{maximum (cardinality) matching} (or shortly, \emph{MCM}),
and a matching that is maximal under inclusion is called a \emph{maximal matching}.
For any parameter $t \ge 1$, a matching that contains at least $1/t$ fraction of the edges in an MCM is called a \emph{$t$-approximate MCM}.
It is easy to see that any maximal matching is a 2-approximate MCM.

A vertex is called \emph{matched} if it is incident on some edge of $M$.
Otherwise it is \emph{free}. For any edge $\{u,v\} \in M$, we say that $u$ (respectively, $v$) is the \emph{mate}
of $v$ (resp., $u$). An \emph{alternating path} is a path whose edges alternate between $M$ and $E \setminus M$.
An \emph{augmenting path} is an alternating path that starts and ends at different free vertices.
It is well-known \cite{HK73} that any matching without augmenting paths of length at most $2k - 3$
is a $(k/(k-1))$-approximate MCM. In particular, if there are no augmenting paths of length at most 3,
we get $3/2$-approximate MCM.

\section{General Graphs}

In this section we present a fully dynamic algorithm for maintaining a maximal matching that is also a $3/2$-approximate MCM.
Our data structure is deterministic
and requires a worst-case update time of $O(\sqrt{n+m})$, for general $n$-vertex graphs with $m$ edges.
Denote the (static) vertex set of the graph by $V=\{1,2\dots,n\}$, and assume for simplicity that $\sqrt{n}$ is an integer. Let $\cG=(G_0,G_1,\ldots) $ be the given sequence of graphs, we assume that the initial graph $G_0$ is empty, and each graph $G_i$ is obtained from the previous graph $G_{i-1}$ by either adding or deleting a single edge.
For each time step $i$, write $G_i = (V,E_i)$, $m_i = |E_i|$.
We maintain the number of edges in the current graph $G = (V,E)$ in the variable $m$.

\subsection{Data Structures}

The algorithm will maintain the following data structures.
\begin{itemize}
\item
The current matching $M$ is stored in an AVL tree. It supports \texttt{insert} and \texttt{delete} in $O(\log n)$ time. Every vertex $v\in V$ holds a $\mate(v)$ that returns its current mate in the matching (or $\bot$ if $v$ is free).

\item
For each vertex $v\in V$ an AVL tree $N(v)$ that stores its current neighbors, and a variable $\deg(v)$ for its degree. It supports \texttt{insert} and \texttt{delete} in $O(\log n)$ time, and extracting arbitrary $r$ neighbors in $O(r)$ time (by traversing the tree).

\item
For each vertex $v\in V$ a data structure $F(v)$ that holds its free neighbors, and supports the following operations: \texttt{insert} and \texttt{delete} in $O(1)$ time, $\free(v)$ that returns TRUE if $v$ has a free neighbor in $O(1)$ time, and $\gfree(v)$ that returns an arbitrary free neighbor of $v$ in $O(\sqrt{n})$ time.

In order to implement $F(v)$ for each vertex $v\in V$, we use a boolean array of size $n$ indicating the current free neighbors, a counter array of size $\sqrt{n}$ that has in position $j$ the number of free neighbors in the range $[\sqrt{n}\cdot j+1,\sqrt{n}(j+1)]$, and a variable for the total number of free neighbors. Now \texttt{insert}, \texttt{delete} and $\free(v)$ are clearly $O(1)$ operations, and in order to implement $\gfree(v)$, we can scan in $\sqrt{n}$ time the counter array for a positive entry, and check the appropriate range in the boolean array.

\item
A maximum heap $F_{max}$ of all free vertices indexed by their degree. It supports \texttt{insert}, \texttt{delete}, \texttt{update-key}, and \texttt{find-max} in $O(\log n)$ time.
\end{itemize}


\begin{figure}[!ht]
\fbox{
\begin{minipage}[t]{165mm}
\texttt{handle-addition}($\{u,v\}$):
\begin{enumerate}
\item Update $N(u)$, $N(v)$, $\deg(u)$, $\deg(v)$ and $F_{max}$;
\item If both $u,v$ are free, $\add(u,v)$;
\item If only $u$ is free: \label{item:ee}
    \begin{enumerate}
    \item Set $v'=\mate(v)$;
    \item Remove $u$ from $F(v')$;
    \item If $\free(v')$:
        \begin{enumerate}
        \item Call $\add(u,v)$; $\add(v',\gfree(v'))$;
        \item Remove $\{v,v'\}$ from $M$;
        \end{enumerate}
    \item Else, for all $w\in N(u)$, add $u$ to $F(w)$;
    \end{enumerate}
\item If only $v$ is free, return to \ref{item:ee}  with roles of $u,v$ switched;
\end{enumerate}
\end{minipage}}\caption{Handle an edge addition.}
\label{fig:handle-add}
\end{figure}


\begin{figure}[!ht]
\fbox{
\begin{minipage}[t]{165mm}
$\add(u,v)$:
\begin{enumerate}
\item Add $\{u,v\}$ to $M$;
\item Remove $u,v$ from $F_{max}$;
\item For $w\in\{u,v\}$:
\begin{enumerate}
\item If $\mate(w)=\bot$, remove $w$ from $F(x)$ \\for all $x\in N(w)$;
\end{enumerate}
\item Set $\mate(u)=v$; $\mate(v)=u$;
\end{enumerate}
\end{minipage}}\caption{Add an edge to the matching.}
\label{fig:add}
\end{figure}


\begin{figure}[!ht]
\fbox{
\begin{minipage}[t]{165mm}
\texttt{\apath}($u$):
\begin{enumerate}
\item For all $w\in N(u)$:
    \begin{enumerate}
    \item Let $w'=\mate(w)$;
    \item If $\free(w')$:
        \begin{enumerate}
        \item Let $x=\gfree(w')$;
        \item break;
        \end{enumerate}
    \end{enumerate}
\item If a free $x$ was found:
    \begin{enumerate}
    \item Call $\add(u,w)$; $\add(w',x)$;
    \item Remove $\{w,w'\}$ from $M$;
    \end{enumerate}
\item Else, if no augmenting path was found:
    \begin{enumerate}
    \item For all $w\in N(u)$ add $u$ to $F(w)$;
    \item Add $u$ to $F_{max}$; Set $\mate(u)=\bot$;
    \end{enumerate}
\end{enumerate}
\end{minipage}}\caption{Finding a length $3$ augmenting path starting at $u$ and adding it to the matching. It is assumed that $\deg(u)\le\sqrt{2n+2m}$ and that $u$ has no free neighbor.}
\label{fig:augmenting}
\end{figure}


\begin{figure}[!ht]
\fbox{
\begin{minipage}[t]{165mm}
\texttt{\surrogate}($u$):
\begin{enumerate}
\item For all $w\in N(u)$:
    \begin{enumerate}
    \item Let $w'=\mate(w)$;
    \item If $\deg(w')\le \sqrt{2m}$, break;
    \end{enumerate}
\item Remove $\{w,w'\}$ from $M$;
\item Call $\add(u,w)$;
\item Return $w'$;
\end{enumerate}
\end{minipage}}\caption{Finding a surrogate low degree vertex for the vertex $u$. It is assumed that $\deg(u)>\sqrt{2m}$ and that $u$ has no free neighbor.}
\label{fig:surrogate}
\end{figure}


\begin{figure}[!ht]
\fbox{
\begin{minipage}[t]{165mm}
\texttt{handle-deletion}($\{u,v\}$):
\begin{enumerate}
\item Update $N(u)$, $N(v)$, $\deg(u)$, $\deg(v)$ and $F_{max}$;
\item If $\{u,v\}\notin M$:
    \begin{enumerate}
    \item If $u$ is free, remove it from $F(v)$; If $v$ is free, remove it from $F(u)$;
    \end{enumerate}
\item Else, if $\{u,v\}\in M$:
    \begin{enumerate}
    \item Remove $\{u,v\}$ from $M$;
    \item For $z\in\{u,v\}$:
        \begin{enumerate}
        \item If $\free(z)$, call $\add(\gfree(z),z)$; \label{item:rr}
        \item Else, if there is no free neighbor for $z$:
            \begin{enumerate}
            \item If $\deg(z)>\sqrt{2m}$, let $z=\surrogate(z)$; Return to \ref{item:rr}.
            \item Else, call $\apath(z)$;
            \end{enumerate}
        \end{enumerate}
    \end{enumerate}
\end{enumerate}
\end{minipage}}\caption{Handle an edge deletion.}
\label{fig:handle-delete}
\end{figure}

\subsection{Algorithm}


At the outset the graph is empty, and we perform an initialization phase for our data structures.
Next, the algorithm is carried out in rounds:
In each round $i = 1,2,\ldots$, a single edge $e_i$ is either added to the graph or deleted
from it, and the algorithm will update the data structure in $O(\sqrt{n + m})$ time to preserve the following invariants at the end of each step $i$.
\begin{invariant}\label{inv:deg}
All free vertices have degree at most $\sqrt{2n+2m}$.
\end{invariant}
\begin{invariant}\label{inv:new-free}
All vertices that became free in round $i$ have degree at most $\sqrt{2m}$.
\end{invariant}
\begin{invariant}\label{inv:match}
The matching $M$ maintained by the algorithm is maximal. Moreover, there are no augmenting paths of length 3 (with respect to $M$).
\end{invariant}

The invariants clearly hold before the first round starts and the edge $e_1$ is handled.
Fix a time step $i$. We will now describe a single round of the algorithm, which handles an edge $e_i$ that is added to the graph or deleted from it.

\subsubsection{Edge Addition}

We start with the case where the edge $e_i = \{u,v\}$ is \emph{added} to the graph, see Figure \ref{fig:handle-add}.
First update the relevant data structures: $N(u),N(v),\deg(u),\deg(v)$ and the keys of $u,v$ in $F_{max}$ (if needed), which takes $O(\log n)$ time.
Next, we distinguish between four cases.
\\\emph{Case 1: Both $u$ and $v$ are matched.} In this case there is nothing to do.
\\\emph{Case 2: Both $u$ and $v$ are free.} In this case we match $u$ with $v$, see Figure \ref{fig:add}. This involves updating the data structures $M$, $F_{max}$ and removing $u,v$ from the free neighbor data structures $F(x)$ for all neighbors $x$ of $u,v$. By Invariant \ref{inv:deg} both $\deg(u)$ and $\deg(v)$ are at most $\sqrt{2n+2m}+1$, so this takes $O(\sqrt{n+m})$ time.
Observe that adding the edge $\{u,v\}$ to $M$ does not create any new augmenting paths
of length at most 3, because by maximality of $M$, both $u,v$ could not have had a free neighbor, and so Invariant \ref{inv:match} is preserved.
\\\emph{Case 3: $u$ is free and $v$ is matched.} In this case, adding the edge $\{u,v\}$ to the graph may give rise to
new augmenting paths of length 3 that include  $\{u,v\}$. Specifically, such a path may exist iff $v'=\mate(v)$ has a free neighbor $w \ne u$. We determine if $v'$ has a free neighbor $w \ne u$ or not in the following way:
First, we remove $u$ from $F(v')$ (we will ``undo'' this before the round is over).
Next, we check if $v'=\mate(v)$ has a free neighbor $w$ (note that $w\neq u$). If we can find one, then we add $\{u,v\}$ and $\{v',w\}$ to $M$, and remove $\{v,v'\}$ from $M$. Observe that when adding the edges we update $F(x)$ \emph{only} for the vertices $x$ that are neighbors of $u$ and $w$ ($v$ and $v'$ were already matched), by Invariant \ref{inv:deg} both $\deg(u)$ and $\deg(w)$ are at most $\sqrt{2n+2m}+1$, so this takes $O(\sqrt{n+m})$ time. If we cannot find such a free neighbor $w$, then $u$ will remain free, and is added to the free neighbor data structures of all its neighbors (in particular to $F(v)$ and if needed to $F(v')$ as well).
\\\emph{Case 4: $u$ is matched and $v$ is free.} This case is symmetric to case 3.

It is easy to see that handling the addition of $\{u,v\}$ in each of the four cases above: (1)
requires an overall update time of $O(\sqrt{n + m})$,  and (2) preserves both Invariants \ref{inv:new-free} and \ref{inv:match}. Invariant \ref{inv:deg} will be handled separately.

\subsubsection{Edge Deletion}

We proceed to the case where the edge $e_i = \{u,v\}$ is \emph{deleted} from the graph, see Figure \ref{fig:handle-delete}.
First update the relevant data structures: $N(u),N(v),\deg(u),\deg(v)$ and $F_{max}$ (if needed), which takes $O(\log n)$ time.
There are two cases to consider.
\\\emph{Case 1: $\{u,v\} \notin M$.} In this case, the only remaining thing to do is to remove $u$ from $F(v)$ if $u$ is free, and remove $v$ from $F(u)$ if $v$ is free.
\\\emph{Case 2: $\{u,v\} \in M$.} Here we delete the edge $\{u,v\}$ from $M$ (we do not update the status of $u$ and $v$ from matched to free yet for technical reasons that will become clear soon, but we will make sure to update $u$ and $v$ with their correct status before the end of this round).
Deleting $\{u,v\}$ from $M$ may give rise to new augmenting
paths of length at most 3 that start at one of the endpoints of this edge.
Next, we show how to handle $u$. The other endpoint $v$ should be handled in the same way.

If $u$ has a free neighbor $w$, we add $\{u,w\}$ to $M$ by calling $\add(u,w)$.
Observe that we \emph{do not} update $F(x)$ for the neighbors $x\in N(u)$ (as $\mate(u)=v$ before $\add$ is called). Since $w$ was free, Invariant \ref{inv:deg} suggests that $\deg(w)\le\sqrt{2n+2(m+1)}$, so this will take only $O(\sqrt{n+m})$ time.
We henceforth assume that $u$ has no free neighbor, and consider two cases.
\\\emph{Case 2.a: $\deg(u) \le \sqrt{2m}$.} In this case we can allow $u$ to become free, but still must search for an augmenting path of length 3 starting at $u$ by calling $\apath(u)$ (see Figure \ref{fig:augmenting}). An augmenting path exists iff some neighbor $w$ of $u$ is matched to $w'$ and $w'$ has a free neighbor $x\neq u$. For each such neighbor $w\in N(u)$ we can in $O(1)$ time detect if its mate $w'$ (recall that $w$ must be matched) has a free neighbor. Only if we find such a $w'$ we do the $O(\sqrt{n})$ operation of actually extracting this free neighbor $x$, and then stop the search (we are guaranteed that $x\neq u$, because we have not changed the status of $u$ to free just yet). If such an $x$ was found, we change $M$ by adding $\{u,w\}$ and $\{w',x\}$ instead of $\{w,w'\}$. Observe that we update $F(y)$ only for neighbors of $x$ (as $u,w,w'$ are recorded as matched), which takes $O(\sqrt{n+m})$ time by Invariant \ref{inv:deg}. If no augmenting path was found, we declare $u$ as a new free vertex (which complies with Invariant \ref{inv:new-free}), and update $F(w)$ for all its neighbors $w$.
A delicate matter that needs attention is the following: if $u$ is the first among $\{u,v\}$ that is handled, $v$ is still recorded as matched, so we will not be able to find an augmenting path of length 3 that starts at $u$ and ends in $v$. However, this path can be detected once we are done with $u$, set its status to free, and handle $v$.
\\\emph{Case 2.b: $\deg(u) > \sqrt{2m}$.}
Note that $u$ cannot become free because its degree is too high, alas it has no free neighbor.
In order to keep $u$ matched, we run $\surrogate(u)$ to find a surrogate $s_u$ for $u$ that
may become free instead of $u$ (see Figure \ref{fig:surrogate}). Even though $\deg(u)$ is high, we claim that after inspecting $\sqrt{2m}$ of the neighbors $w\in N(u)$, we must have found one with degree at most $\sqrt{2m}$, and then stop the scan. Indeed, otherwise the sum of degrees in the graph would be more than $\sqrt{2m}\cdot\sqrt{2m}=2m$ (note that $\mate(w)$ are distinct for different $w$), which is impossible. Since the surrogate $s_u$ has degree at most $\sqrt{2m}$, changing its status to free (if needed) would not violate Invariant \ref{inv:new-free}. Next, handle $s_u$ as $u$ is handled above just before \emph{Case 2.a} (that is, find a free neighbor of $s_u$ or an augmenting path of length 3). Note that handling $s_u$ cannot bring us to \emph{case 2.b}, so there is no risk of an infinite loop. We claim that no augmenting path of length 3 can remain, because any augmenting path emanating from $s_u$ is detected, and the edge $\{u,w\}$ that is added to $M$ in $\surrogate(u)$ cannot be a part of an augmenting path because $u$ has no free neighbors\footnote{Again we exclude augmenting paths ending at $v$, those will be detected later.}.

It is easy to see that handling the deletion of $\{u,v\}$ in each of the two cases above: (1)
requires an overall update time of $O(\sqrt{n + m})$, and (2) preserves both Invariants \ref{inv:new-free} and \ref{inv:match}.

\subsubsection{Bounding the Degree of Free Vertices}

Here we show how to preserve Invariant \ref{inv:deg}, that free vertices have bounded degrees, which is a key property in our algorithm.
The general idea is to identify some \emph{problematic} vertices at the end of each round, and to
\emph{correct} them.
We say that a vertex is \emph{problematic} if (i) it is free, and (ii) its degree exceeds $\sqrt{2m}$.
Such a problematic vertex $x$ is corrected by applying \emph{case 2.b} above on $x$. That is, find a surrogate $s_x$ that may become free instead of $x$, with $\deg(s_x)\le\sqrt{2m}$. In order to preserve Invariant \ref{inv:match}, we then find an augmenting path of length at most $3$ emanating from $s_x$ if one exists.

Since each correction takes $O(\sqrt{n+m})$ time, we can only afford to correct $O(1)$ problematic vertices at the end of each round.
It turns out that correcting the following three vertices (if they are problematic) suffices: first the two endpoints $u$ and $v$ of the handled edge $e_i$,
and afterwards a free vertex $x$ with maximal degree (such a vertex can be extracted from the heap $F_{max}$ in $O(\log n)$ time).

The next lemma implies that Invariant \ref{inv:deg} is preserved.

\begin{lemma}
At the end of each round $i$, for any \emph{free} vertex $x$, $\deg(x) \le \sqrt{2n + 2m}$.
\end{lemma}
\begin{proof}
Recall that $G_i = (V,E_i)$ denotes the $i$-th graph in the graph sequence $\cG$,
and $m_i = |E_i|$ stands for the number of edges in it.
First observe that the degree of a problematic vertex $x$ cannot change as long as it is problematic. This is because any change to the degree would mean that we added or deleted an edge touching $x$, and so must have corrected it.

Seeking contradiction, assume that at round $t$ the vertex $x$ is free and $\deg(x)>\sqrt{2n + 2m_t}$. Let $k<t$ be such that at round $k$, $\deg(x)\le\sqrt{2m_k}$ and in all rounds $k<j\le t$ we have that $\deg(x)>\sqrt{2m_j}$. Since $x$ is problematic in all the rounds from $k+1$ to $t$, its degree does not change, and it follows that $\sqrt{2m_k}>\sqrt{2n+2m_t}$, or $m_k-m_t>n$.
Let $k\le q<t$ be the minimal round such that the number of edges in every round $q+1,\dots,t$ is less than $m_k$, observe that $n<t-q$ because there must have been more that $n$ deletions of edges.

We claim that $x$ must be corrected in one of these $n$ rounds. To prove this, it suffices to show that every vertex that becomes problematic after round $q$ will have smaller degree than $\deg(x)$. Once it becomes problematic its degree cannot change, so in $F_{max}$ the vertex $x$ will be handled before all the ``new'' problematic vertices. Since we handle one vertex from $F_{max}$ at each round, and there are more than $n$ rounds from $q$ to $t$, it must be that $x$ is handled in one of them. Suppose vertex $w$ becomes problematic at the conclusion of round $q<r\le t$. There could be three reasons for this: First, if the degree of $w$ changed in round $r$, then actually it must have been corrected (recall that we correct both endpoints of the new edge). Second, if $w$ became a new free vertex, then by Invariant \ref{inv:new-free} $\deg(w)\le\sqrt{2m_r}$, and as $x$ is problematic in round $r$, $\deg(x)>\sqrt{2m_r}$. The last case is that $w$ was already free, and the number of edges decreased. But then in round $r-1$, $x$ is problematic and $w$ is not, thus $\deg(w)\le\sqrt{2m_{r-1}}<\deg(x)$. We conclude that $x$ is indeed corrected before round $t$ ends, a contradiction.
\qed
\end{proof}

We have shown the following.
\begin{theorem}
Starting with the empty graph on $n$ vertices, a maximal matching in the graph which is also a $3/2$-approximate MCM can be maintained in time $O(\sqrt{n+m})$ per edge update, where $m$ is the (current) number of edges.
\end{theorem}

\section{Low Arboricity Graphs}\label{sec:arb}

In this section we consider graphs with arboricity bounded by $c$.
\begin{definition}\label{def:arb}
A graph $G=(V,E)$ has arboricity $c$ if
\[
c=\max_{U\subseteq V}\left\lceil\frac{|E(U)|}{|U|-1}\right\rceil,
\]
where $E(U)=\left\{\{u,v\}\in E\mid u,v\in U\right\}$.
\end{definition}
The family of graphs with bounded arboricity is the family of uniformly \emph{sparse} graphs. In particular, it contains planar and bounded genus graphs, bounded tree-width graphs and in general all graphs excluding fixed minors.

A $\Delta$-orientation of an undirected graph $G=(V,E)$ is a directed graph $H=(V,A)$ where $A$ contains the same edges as in $E$ (each edge is given a direction), so that the out-degree of every vertex in $H$ is at most $\Delta$.
A well known theorem of Nash and Williams \cite{NW64} asserts that a graph has arboricity at most $c$ iff $E$ can be partitioned to $E_1,\dots,E_c$ such that $(V,E_i)$ is a forest for all $1\le i \le c$. This suggests that one can select an arbitrary root for all trees in the forests, and direct all edges towards the root. The out-degree of any vertex in each forest $(V,E_i)$ is at most $1$, so $G$ has a $c$-orientation.

Consider a sequence of graphs ${\cal G} = (G_0,G_1,\dots,G_k)$ on the vertex set $V$ with $|V|=n$.  We say that ${\cal G}$ has arboricity $c$ if: $G_0$ is the empty graph, for all $1\le i \le k$, $G_i$ is obtained from $G_{i-1}$ by adding or deleting an edge, and all graphs $G_i$ have arboricity at most $c$. We say that an algorithm maintains a $\Delta$-orientation for ${\cal G}$ with amortized time $T$ if it provides a $\Delta$-orientation $H_i$ for every $G_i$, and the total number of edge re-orientations is $k\cdot T$.
Brodal and Fagerberg \cite{BF99} showed the following theorem:

\begin{theorem}[Brodal and Fagerberg \cite{BF99}]\label{thm:BF}
For any sequence of graphs $\cG$ with arboricity $c$, and any $\Delta\ge 2\delta>2c$ there is an algorithm that maintains a $\Delta$-orientation for $\cG$ with amortized time $T=O(\Delta+\frac{\Delta+1}{\Delta+1-2\delta}\cdot\log_{\delta/c}n)$.
\end{theorem}

\subsection{Reduction from Matchings to Orientations}

Now we explain how to use an algorithm ${\cal A}$ that maintains a $\Delta$-orientation in amortized time $T$, in order to obtain an algorithm that maintains a maximal matching in amortized time $O(\Delta+T)$. The idea behind the algorithm is the following: every vertex is responsible to notify about its state - free or matched - all the vertices it is pointing to (there are at most $\Delta$ such vertices). In other words, each vertex knows exactly who is free among the (possibly many) vertices pointing towards it, but knows nothing of the (at most $\Delta$) vertices it is pointing to. This partial information enables vertices to pay only $O(\Delta)$ time in order to retrieve all information about their neighbors, and also they can perform the necessary status updates in $O(\Delta)$ time.

Consider a sequence of graphs ${\cal G}$ with arboricity $c$. For every graph $G_i\in {\cal G}$ we have an orientation $H_i$ given by algorithm ${\cal A}$. For a vertex $u\in V$ denote by $N_i(u)$ the set of neighbors of $u$ in $G_i$, and let $D_i(u)\subseteq N_i(u)$ be the set of vertices such that the edge $(u,v)$ is directed out of $u$ in the current orientation induced by $H_i$. Observe that $|D_i(u)|\le\Delta$ for all $u\in V$ and $0\le i\le k$. We will maintain the following data structures:
\begin{itemize}
\item

A data structure $D(u)$ holding all the vertices of $D_i(u)$.

\item

A data structure $F(u)$ holding all the \emph{free} vertices $v\in N_i(u)\setminus D_i(u)$.

\item

A data structure $M$ containing all the matched edges, and a value for every vertex indicating whether it is free or matched.

\end{itemize}

Each of these data structures will be implemented using an array of size $n$ augmented with a linked list. For $D(u)$ and $F(u)$ the array will be boolean, while for $M$ entry $i$ will include the mate of node $i$ in the matching, or $\bot$ if node $i$ is free. Insertions can easily be done in $O(1)$ time, however deletions are done only in the array - the linked lists may contain extra elements that are in fact deleted. Whenever an extraction is needed, we go over the linked list starting at its head, and verify every element we encounter against the array. If an extraction took $r$ verifications until a valid element was found, then we have deleted $r-1$ elements from the list and the cost is divided among these $r-1$ delete operations. We conclude that all these data structures support insertion, deletion and extraction in amortized $O(1)$ time.

We now give an overview of the algorithm that maintains a maximal matching in $\cG$.
\begin{itemize}
\item

{\bf Orientation:}
At every step we run algorithm ${\cal A}$ that preserves a $\Delta$-orientation for the current graph. We then update the data structures $F$ and $D$ as described in Figure \ref{fig:update} so that they will be consistent with the current orientation. If there are $t$ edge re-orientations then this takes $O(t)$ time.

\item

{\bf Insertion:} (see Figure \ref{fig:insert}).
Upon edge $\{u,v\}$ insertion, we check if we can add the edge to the matching $M$, and if so remove $u$ and $v$ from the free neighbors data structure $F(w)$ for every point $w$ in $D(u)$ and $D(v)$. As all out-degrees are at most $\Delta$, this takes $O(\Delta)$ time.

\item

{\bf Deletion:} (see Figure \ref{fig:delete}).
Upon deletion of a non-matched edge $\{u,v\}$, nothing more needs to be done. The interesting case is deleting a matched edge $\{u,v\}$: here we try to find new match for $u$ (respectively $v$) by going over $F(u)$ and $D(u)$ (respectively $F(v)$ and $D(v)$). Observe that $D(u)$ may contain both free and matched vertices, so we must go over all of its elements, which takes $O(\Delta)$ time. Although $F(u)$ may be very large, we are guaranteed that all vertices in it are free, so we just need to extract one of them (or tell that $F(u)$ is empty) which takes amortized $O(1)$ time.
\end{itemize}

We conclude that the following result holds.

\begin{theorem}\label{thm:arb}
For every sequence of graphs ${\cal G}$ on $n$ vertices with arboricity $c$, and for any $\Delta>2c$, there is an algorithm that maintains a maximal matching in amortized time $T=O(\Delta+\log_{\Delta/c}n)$.
\end{theorem}

\begin{proof}
By Theorem \ref{thm:BF} with parameters $\delta=2c$ and $\Delta = 5c$ we can maintain an orientation with out-degree at most $\Delta$ in amortized $T$ time. Our algorithm in addition performs at most $O(\Delta +t_i)$ operations, where $t_i$ is the number of edge re-orientations at step $i$. We conclude that the amortized time is $O(T)$. It can be easily verified that the data structures $F$ and $D$ are consistent throughout the execution of the algorithm, and that $M$ is indeed a maximal matching.
\end{proof}
\begin{corollary}
Choosing $\Delta = 6c+\frac{\log n}{\log((\log n)/c)}$ will give amortized time $T=O\left(\frac{\log n}{\log((\log n)/c)} + c\right)$, for any $c= o(\log n)$.
In particular, we will get amortized time $O(\log n/\log\log n)$ when $c\le\log^{1-\epsilon}n$ (for any constant $\epsilon>0$).
For any $c  = \Omega(\log n)$, the amortized time is $T = O(c)$.
\end{corollary}

\begin{figure}[!ht]
\fbox{
\begin{minipage}[t]{165mm}
\texttt{update-orientation}($\{u,v\}$):
\begin{enumerate}
\item Run algorithm ${\cal A}$:
\item If edge $(u,v)$ was inserted and directed from $u$ to $v$:
    \begin{enumerate}
    \item Add $v$ to $D(u)$; If $u$ is free, add $u$ to $F(v)$;
    \end{enumerate}
\item Otherwise, if edge $(u,v)$ was deleted:
    \begin{enumerate}
    \item Remove $v$ from $D(u)$; If $u$ is free, remove $u$ from $F(v)$;
    \end{enumerate}
\item For every edge $(x,y)$ that was re-oriented to $(y,x)$:
    \begin{enumerate}
    \item Delete $y$ from $D(x)$, and add $x$ to $D(y)$;
    \item If $x$ is free, delete $x$ from $F(y)$; If $y$ is free, add $y$ to $F(x)$;
    \end{enumerate}
\end{enumerate}
\end{minipage}}\caption{Updating the orientation for bounded arboricity graphs}
\label{fig:update}
\end{figure}
\begin{figure}[!ht]
\fbox{
\begin{minipage}[t]{165mm}
\texttt{insert \{u,v\}}:
\begin{enumerate}
\item Run \texttt{update-orientation}($\{u,v\}$);
\item If both $u,v$ are free:
    \begin{enumerate}
    \item $M=M\cup \{u,v\}$;
    \item For each $w\in D(u)\cup D(v)$, remove $u$ and $v$ from $F(w)$;
    \end{enumerate}
\end{enumerate}
\end{minipage}}\caption{Edge insertion for bounded arboricity graphs}
\label{fig:insert}
\end{figure}
\begin{figure}[!ht]
\fbox{
\begin{minipage}[t]{165mm}
\texttt{delete \{u,v\}}:
\begin{enumerate}
\item Run \texttt{update-orientation}($\{u,v\}$);
\item If $\{u,v\}\in M$, then $M=M\setminus\{u,v\}$, and for $w\in\{u,v\}$:
    \begin{enumerate}
    \item If $F(w)\neq\emptyset$, find some $x\in F(w)$ and $M=M\cup\{x,w\}$;
    \item Otherwise, for each $x\in D(w)$:
        \begin{enumerate}
        \item If $x$ is free:
            \begin{enumerate}
            \item set $M=M\cup\{x,w\}$;
            \item break;
            \end{enumerate}
        \end{enumerate}
    \item If $\{x,w\}$ was added to $M$, then for all $y\in D(w)\cup D(x)$, remove $w$ and $x$ from $F(y)$;
    \end{enumerate}
\end{enumerate}
\end{minipage}}\caption{Edge deletion for bounded arboricity graphs}
\label{fig:delete}
\end{figure}
\vspace{0.12in}

\section{Maximal Matching Using Optimal Space}

In this section we show that the algorithm of Section 4 for bounded arboricity graphs can be implemented using only $O(n+m)$ space. Using the fact that any graph on $m$ edges has arboricity at most $\sqrt{m}$ (see \cite{DHS91}), we conclude that the amortized update time is $O(\sqrt{m})$ even for \emph{arbitrary} graphs.\footnote{Again we ignore additive terms that depend at most logarithmically on $n$.}

We remark that using dynamic hash tables it would have been easy to obtain space $O(n+m)$, but we desire a fully deterministic algorithm. Recall that $c$ is the maximum arboricity of the graph, and let $\Delta=5c$ be the maximum allowed out-degree in the graph.

\vspace{0.08in}
\noindent {\bf Data Structures:}
A data structure $M$ containing all the matched edges (and a value for each vertex
indicating whether it is free or matched) will be maintained just as in Section 4.
In addition, we will maintain for each vertex $u$ its current neighbors $N(u)$, its outgoing neighbors $D(u)$ and its free incoming neighbors $F(u)$ as linked lists (and maintain their sizes as well). We will also maintain the degree $\deg(u)$ and a variable indicating whether $u$ is free or matched. The lists $N(u)$ and $F(u)$ will not be authentic at all times, in a sense that they may contain redundant elements.
In contrast, the $D(u)$ lists will always be authentic. In fact, the Brodal-Fagerberg algorithm \cite{BF99} that we use maintains these lists explicitly in $O(n+m)$ space, so we may assume that $D(u)$ is always up to date. Also, $\deg(u)$ will contain the current degree of $u$, and thus it may be smaller than $|N(u)|$ at some stages throughout the execution of the algorithm.
In order to control the total space used, we will guarantee that $|N(u)|$ (respectively, $|F(u)|$) never exceeds $\deg(u)$ by more than a factor of $2$ (resp., $3$). We will also use a "smart"
boolean array of size $n$ (an array that allows to reset all the elements to $0$ in $O(1)$ time, see, e.g., \cite{AHU74}) for authentication of the lists. To meet the desired space requirement, the
same smart array will be re-used for all lists.

\vspace{0.08in}
\noindent {\bf Handling the $N(u)$ lists:}
First note that these lists are never used by the algorithm, their sole purpose is to assist in the authentication of the $F(u)$ lists.
Whenever an edge $\{u,v\}$ is added to the graph, we simply add $u$ to $N(v)$ and $v$ to $N(u)$ in $O(1)$ time. However, upon deletion of the edge $\{u,v\}$, we put $10\Delta$ tokens on the edge to be used later for authenticating $N(u),F(u)$ and $N(v),F(v)$. As there is at most one edge deletion per round, we can afford
to spend so many tokens. For every change to $N(u)$, we check that $|N(u)| < 2\deg(u)$. If it is not, we authenticate it in the following manner. Iterate over the list $N(u)$, and for every element $w \in N(u)$:
\begin{itemize}
\item Search $D(u)$ for $w$ and search $D(w)$ for $u$, if none of them was found, remove $w$ from $N(u)$.
\end{itemize}
As $|D(u)|,|D(w)| \le \Delta$, the cost of this authentication is at most $2\Delta\cdot |N(u)|=4\Delta\deg(u)$.
Observe that at least $\deg(u)$ elements are removed from $N(u)$, each due to a deleted edge. These deleted edges can contribute $4\Delta$ tokens each for this authentication of $N(u)$, which suffices to cover the cost (another $4\Delta$ tokens will be used for $N(v)$, where $\{u,v\}$ was the deleted edge, and $\Delta$ tokens each for $F(u),F(v)$ as described below).

\vspace{0.08in}
\noindent {\bf Handling the $F(u)$ lists:}
For $F(u)$, we will keep adding elements in $O(1)$ time according to the algorithm, but similarly to $N(u)$, deletions are postponed. We will place $O(1)$ tokens for each delete operation on some list $F(u)$ that is postponed. Spending these $O(1)$ extra tokens for each operation should increase the total amortized time
by at most a constant factor. Similarly to $N(u)$, $F(u)$ may contain non-authentic elements. Note that unlike $N(u)$, we have only $O(1)$ tokens per deleted element, thus intuitively, we must authenticate $F(u)$ using only $O(1)$ time per element. Next we show how to make sure that the size of $F(u)$ is never more than $3\deg(u)$, and how to extract an authentic element from $F(u)$.

First we show how to extract an authentic element: reset the smart array, and update it to contain the (authentic) elements of $D(u)$ in $O(c)$ time. Then we will traverse the $F(u)$ list and upon encountering element $w$ do the following:
\begin{itemize}
\item If $w$ is matched, remove $w$ from $F(u)$ (using the $O(1)$ tokens created by the deletion of $w$ from $F(u)$).
\item Otherwise, if $w\in D(u)$, remove $w$ from $F(u)$. (Checking takes $O(1)$ time using the array, so we can pay for it with the $O(1)$ tokens).
\item Otherwise, search for $u$ in $D(w)$:
\begin{itemize}
\item If $u$ is found, then $w$ is an authentic free incoming neighbor of $u$. We spent $\Delta$ time, but can terminate the search.
\item If $u$ was not found, then the edge $\{u,w\}$ must have been deleted. Thus we remove $w$ from $F(u)$.
We spent $\Delta$ time, which can be payed for by $\Delta$ of the tokens placed on the deleted edge $\{u,w\}$ exactly for this purpose.
\end{itemize}
\end{itemize}
Each of the non-authentic elements we encountered had enough tokens for executing its removal, thus the actual cost of the search is only $\Delta$. Observe that the algorithm requires at most two extractions of authentic elements at every round (when edge $\{u,v\}$ is deleted, extractions are required from $F(u)$ and $F(v)$), and thus we can afford to spend $\Delta = O(c)$ time for each.

Finally, we show how to control the size of $F(u)$.
Whenever an element is added or removed from $F(u)$ we check if $|F(u)|\ge 3\deg(u)$, and if so start a (partial) authentication process that will use tokens in order to reduce its size down to at most $2\deg(u)$.
We reset the smart array, and initialize it with $N(u)$. As $|N(u)|< 2\deg(u)$, this will take at most $O(\deg(u))$ time. Recall that it could be that $N(u)$ is not authentic. Now iterate over $F(u)$, and for each element $w \in F(u)$:
\begin{itemize}
\item If $w$ is matched, remove $w$ from $F(u)$.
\item Otherwise, if $w\notin N(u)$, remove $w$ from $F(u)$. (This takes $O(1)$ time using the array).
\end{itemize}
This (partial) authentication process may make one-sided mistakes: elements that should have been deleted from $F(u)$ may still remain, but no element will be deleted from $F(u)$ un-necessarily, since $N(u)$ contains all the authentic neighbors of $u$. Observe that since $|F(u)|\ge 3\deg(u)$ but $|N(u)|\le 2\deg(u)$, at least $\deg(u)$ elements must have been removed from $F(u)$ in the process (because if $w\notin N(u)$ it will be removed from $F(u)$). Each removed element had $O(1)$ tokens for its removal, so we had the $O(\deg(u))$ tokens to pay for this authentication process.
\vspace{0.08in}

To conclude, we have shown that the algorithm from Section \ref{sec:arb} can be implemented in optimal space $O(n+m)$, without increasing the amortized update time by more than a constant factor. To obtain maximal matching in amortized time $O(\sqrt{m})$ and within space $O(n+m)$, we do the following.
In order to use the orientation algorithm, at the beginning of every stage we set a bound on the maximum arboricity to be $c=2\sqrt{m}$. Whenever $m$ changes by a factor of $2$ (which happens in the worst case every $\sqrt{m}/2$ rounds), we end the current stage, reset the value of $c$, and recompute the orientation (this can be done in $O(m)$ time--which increases the amortized time only by a constant, and within $O(n+m)$ space). Recall that $\Delta=5c\le 20\sqrt{m}$, so the amortized time per update is indeed $O(\sqrt{m})$. (The formal proof is deferred to the full version.)
Thus we have the following theorem.
\begin{theorem}
Starting with the empty graph on $n$ vertices, a maximal matching can be maintained in amortized time $O(\sqrt{m})$ using space $O(n+m)$, where $m$ is the (current) number of edges.
\end{theorem}

\vspace{0.01in} {\bf Acknowledgements.~}
We thank Tsvi Kopelowitz and Itay Gonshorovitz for helpful discussions.


\end {document}